\documentclass{llncs}

\usepackage{amssymb,amsmath}
\usepackage{color}
\usepackage[T1]{fontenc}

\newcommand{\W}[1][XXX]{\normalfont{W[#1]}}
\newcommand{\NP}{\normalfont{NP}}
\newcommand{\Poly}{\normalfont{P}}
\newcommand{\FPT}{\normalfont{FPT}}
\newcommand{\NPO}{\normalfont{NPO}}

\newcommand{\Card}[1]{|#1|}
\newcommand{\Yes}{\textsc{Yes}}
\newcommand{\No}{\textsc{No}}

\newcommand{\pcproblem}[4]{\begin{samepage}\begin{quote} \textsc{#1}\\ 
\textit{Instance:} #2\\ \textit{Parameter:} #3\\ \textit{Question:} #4 \end{quote}\end{samepage}}
\newcommand{\pcaproblem}[4]{\begin{samepage}\begin{quote} \textsc{#1}\\ 
\textit{Instance:} #2\\ \textit{Parameter:} #3\\ \textit{Output:} #4 \end{quote}\end{samepage}}

\definecolor{lightyellow}{cmyk}{0,0,.7,0}

\begin{document}

\institute{SITE, University of Ottawa, 800 King Edward, Ottawa, Ontario K1N 6N5, Canada. \email{casteig@site.uottawa.ca}
\and Department of Computing, Macquarie University, Sydney, NSW 2109, Australia. \email{\textbraceleft{}bernard.mans,luke.mathieson\textbraceright{}@mq.edu.au}}

\title{On the Feasibility of Maintenance Algorithms in Dynamic Graphs}
\author{Arnaud Casteigts\inst{1} \and Bernard Mans\inst{2} \and Luke Mathieson\inst{2}}
\date{}
\maketitle

\begin{abstract} 
Near ubiquitous mobile computing has led to intense interest in dynamic graph theory. This provides a new and challenging setting for algorithmics and complexity theory. For any graph-based problem, the rapid evolution of a (possibly disconnected) graph over time naturally leads to the important complexity question: is it better to calculate a new solution from scratch or to adapt the known solution on the prior graph to quickly provide a solution of guaranteed quality for the changed graph?

In this paper, we demonstrate that the former is the best approach in some cases, but that there are cases where the latter is feasible. We prove that $\W[1]$-hardness for the parameterized approximation problem implies the non-existence of a maintenance algorithm for the given approximation ratio, even given time exponential in the solution size and polynomial in the over all size --- i.e., even with a large amount of time, having a solution to a very similar graph does not help in computing a solution to the current graph. To achieve this, we formalize the idea as a \emph{maintenance algorithm}. To illustrate our results we show that \textsc{$r$-Regular Subgraph} is $\W[1]$-hard for the parameterized approximation problem and thus has no maintenance algorithm for the given approximation ratio. Conversely we show that \textsc{Vertex Cover}, which is fixed-parameter tractable, has a $2$-approximate maintenance algorithm. The implications of $\NP{}$-hardness and $\NPO{}$-hardness are also explored.
\end{abstract}

\section{Introduction}

\sloppypar With the development of sufficiently small and powerful hardware, mobile computing devices present a new and interesting network environment in which the complexity and algorithmics for even well understood problems can change radically. To model this sort of network, various notions of dynamic graphs have been developed, including \emph{delay-tolerant}~\cite{Fall03}, \emph{disruptive-tolerant}, \emph{intermittently-connected}~\cite{Zhang06}, \emph{opportunistic}, \emph{time-varying}~\cite{CasteigtsFMS10} and \emph{evolving}~\cite{XuanFerreiraJarry03,CCF12}. Each of these models differing aspects of the dynamics under differing assumptions or with different applications in mind. One of the key aspects of dynamic graphs is that a graph may not be connected at any given moment, or ever, but due to the appearance and disappearance of edges, it may still be possible to construct journeys (rather than paths) through the graph over time and space. However, dependent on the assumptions of the system, there may be even fundamental properties or traditional problems which are incomputable or definitionally ambiguous (e.g.,~\cite{XuanFerreiraJarry03}). 
The results presented here are independent of the model as long as the change in the graph occurs (or can be described) in a discrete manner.

Given a traditional graph problem $\Pi$ and a sequence of graphs $\{G_i\}$, we may reframe $\Pi$ in several ways as a dynamic graph problem; \emph{permanent} --- where we try to find a solution that holds at all points in time (e.g., a single set of dominators that cover all nodes in every $G_i$), \emph{over-time} --- where the solution is defined with respect to the sequence as a whole (e.g., a set of dominators such that each node is covered in at least one $G_i$), and \emph{evolving} --- where we compute a (possibly different) solution for each $G_i$. 
Permanent and over-time problems may appear somewhat less general in that they require prior knowledge of the graph dynamics.
However it is worth noting that there are relevant in few, yet important, practical scenarios, such as with a known schedule 
 (e.g., public transports, low-earth satellites, sensors with sleeping schedule), or a known schedule property (e.g., \emph{periodicity}~\cite{FMS09} or \emph{recurrence}~\cite{CasteigtsFMS10} --- i.e., edges that exist once are guaranteed to re-appear at some known, bounded, or unbounded time). 

It can be observed that for some problems, e.g., {\em covering} problems like {\sc Vertex Cover}, all three variants are strongly related. Given such a problem $\Pi$, one can easily check that a solution to {\em permanent} $\Pi$ is a (possibly far from optimal) solution to {\em evolving} $\Pi$, and a solution to any $G_i$ in the {\em evolving} $\Pi$ is also valid for {\em over-time} $\Pi$. The connexion is even stronger: the {\em intersection} of solutions to {\em evolving} $\Pi$ is valid for {\em over-time} $\Pi$, and their {\em union} is a solution to {\em permanent} $\Pi$. From this perspective, the {\em evolving} variant appears quite central, and the 
\emph{permanent} or \emph{over-time} $\Pi$ actually form upper and lower bounds for each solution in the  \emph{evolving} $\Pi$.
In this paper, we focus on the \emph{evolving} $\Pi$ and look at finding the solution for each point in time.

The dynamic context is an interesting algorithmic and complexity setting. Although a na\"{i}ve approach would simply compute a new solution for every $G_i$, it is possible that given a relatively limited amount of change in the graph at each step we may leverage the previous solution to allow quick computation of the new solution. This leads to the idea of a \emph{maintenance algorithm} (q.v., Section~\ref{sect:feasibility}). Moreover this approach fits well with the real world inspiration, where along with a possibly rapid pace of change, there is often no centralised control. If the previous solution can be used to compute the new solution, this suggests a certain localization, which would be ideal for mobile infrastructureless networks. 
Considering that dynamic graphs can be permanently disconnected (while still offering connectedness over time), we initially consider well-known problems that are not ambiguous or undefined in a disconnected setting: \textsc{$r$-Regular Subgraph} and \textsc{Vertex Cover}. A solution must be well defined for each (possibly) partitioned component and trivial solutions must exist for trivial cases (e.g., single nodes).

As we are mainly interested in maintaining a solution of some guaranteed quality for each particular problem,
our results can be summarized as follows:
\begin{itemize}
\item we prove that $\W[1]$-hardness for the parameterized approximation problem implies the non-existence of a maintenance algorithm for the given approximation ratio, even given time exponential in the solution size and polynomial in the over all size,
\item we show that \textsc{$r$-Regular Subgraph} and several variants have no parameterized approximation algorithms thus have no $\FPT{}$-time maintenance algorithm, and conversely,
\item we show that problems that are fixed-parameter tractable may have an approximate maintenance algorithm by proving that this is indeed the case for \textsc{Vertex Cover},
\item we establish a complexity classification for maintenance algorithms that provides a strong relationship with parameterized approximation complexity~\cite{Marx08}.
\end{itemize}

Similar ideas regarding dynamic algorithms and complexity have been explored before, but from notably different positions. Typically previous work has focussed on polynomial or logarithmic time problems, and did not include the consideration of locality that is central to many practical applications for dynamic complexity (central oversight and communication is neither guaranteed nor generally desirable in a dynamic network), Section~\ref{sect:feasibility} gives the relevant definitions. However many interesting results and ideas are applicable in this context. Patnaik \& Immerman~\cite{PatnaikImmerman97} consider dynamic complexity from a descriptive complexity theory perspective, defining \textsc{DynFO}, a class of dynamic problems that are expressible in first order logic. Weber \& Schwentick~\cite{WeberSchwentick07} build upon this, again concentrating on a descriptive complexity approach. Holm, de Lichtenberg \& Thorup~\cite{HolmLichtenbergThorup01} give a series of results that can readily be interpreted as maintenance algorithms in our context. Their results for \textsc{Connectivity}, \textsc{2-Edge} and \textsc{Biconnectivity} rely on the maintenance under edge deletion and addition of a solution for \textsc{Minimum Spanning Forest}, giving polylogarithmic running times for all problems, but with no bound on locality. Miltersen \emph{et al.}~\cite{Miltersen1994} present another, similar approach where the dynamism is achieved at a lower level by perturbing individual bits in the input. They also focus on problems of polynomial complexity showing, in our context, that problems such as the \textsc{Circuit Value Problem} and \textsc{Propositional Horn Satisfiability} have no polylogarithmic maintenance algorithms but that interestingly there exist other $\Poly{}$-complete problems that do. Ausiello, Bonifaci \& Escoffier~\cite{AusielloBonifaciEscoffier11} discuss a different model, where there is only interest in using an existing optimal solution to solve (to some degree of approximation) a perturbed instance (the model is generally called \emph{reoptimization}). This approach is in some ways an unsuitable perspective for our context, as we do not expect at any point to be able to compute an optimal solution, however their negative results in particular carry to our setting, such as \textsc{Min Coloring} having no $\Poly{}$-time maintenance algorithm for approximation ratio $\frac{4}{3}-\varepsilon$ for any $\varepsilon$.

The rest of the paper is organized as follows. We define our notation and provide our main results in Section~\ref{sect:feasibility}. We provide some basic and necessary background on parameterized complexity and parameterized approximation in Section~\ref{sect:para}. We provide the main proofs of our results and of results on parameterized approximation complexity for \textsc{Vertex Cover}, \textsc{$r$-Regular Subgraph} and some generalizations of \textsc{$r$-Regular Subgraph} in Section~\ref{sec:evolving_graphs}. We give some concluding remarks in Section~\ref{sect:conclusion}.

\section{The Feasibility of Maintenance}
\label{sect:feasibility}

The key difference in the dynamic setting with regards to computation is that as the graph is continually changing, we are required to continually recompute solutions for the problem. Hence the style of algorithm we are concerned with is somewhat different. In this case we would ideally like an algorithm that can exploit the known solution to the previous graph, which is somewhat similar to the current graph, to produce a solution for the current graph that maintains a desired solution quality (whether optimal or within some approximation bound) in a time faster than it would take to completely recompute the solution. We formalize the idea as a \emph{maintenance algorithm}.

\begin{definition}[Maintenance Algorithm]
Given a sequence of graphs $\{G_{i}\}$ where the editing distance between $G_{i}$ and $G_{i+1}$ is $1$, $\mathcal{A}$ is a maintenance algorithm for problem $\Pi$ if given a solution $S_{i}$ for $\Pi$ on graph $G_{i}$, $\mathcal{A}$ can compute a solution $S_{i+1}$ for $\Pi$ on graph $G_{i+1}$ that preserves the quality of the solution.
\end{definition}

The \emph{editing distance} between two graphs is the number of \emph{editing operations} that have to be performed to obtain one from the other. Typically the editing operations we will be interested in are some subset of edge addition, edge deletion, vertex addition and vertex deletion. We set the editing distance here to $1$ in order to remain sufficiently general in establishing \emph{negative} results.

As the graph changes relatively quickly, we would like the computation to be done within a limited computation time (preferably a small constant). The notion of a \emph{bounded} maintenance algorithm formalizes this bound.

\begin{definition}[Bounded Maintenance Algorithm]
A maintenance algorithm is bounded if there is a function $f$ such that at each step, the computation performed at each vertex is time bounded by $f$. We refer to such an algorithm as a $f$-maintenance algorithm.
\end{definition}

Also, because in a practical setting the nodes model independent devices, we would like the computation to be done as locally as possible, {\it i.e.,} the state of each vertex being based only on the state of its $r$-neighborhood, with $r \ll diameter(G)$.

\begin{definition}[$r$-Local Maintenance Algorithm]
A maintenance algorithm is $r$-local if the state of each vertex is computed based only on the states of vertices within distance $r$. We refer to such an algorithm as a $r$-local maintenance algorithm.
\end{definition}

For example a $1$-local $O(1)$-maintenance algorithm allows each vertex to perform a constant number of computational steps at each iteration of the algorithm, knowing only the states of its direct neighbors. It is with these computational restrictions that knowledge of a prior solution and the concept of a maintenance algorithm can be most useful. Indeed there are many prior results that fit into this framework, though previously there has been little consideration of the locality. For example,

\begin{lemma}[\cite{Miltersen1994}]
\textsc{Undirected Forest Accessibility} has a $d$-local $O(\log n)$-maintenance algorithm where $d$ is the diameter of the graph and $n$ is the number of vertices.
\end{lemma}

\begin{lemma}[\cite{HolmLichtenbergThorup01}]
\textsc{Minimum Spanning Forest} has a $d$-local $O(\log^{4}n)$-maintenance algorithm where $d$ is the diameter of the graph.
\end{lemma}

\begin{lemma}[\cite{RodittyZwick04}]
\textsc{Directed Reachibility} has a $d$-local $O(m+n\log n)$-maintenance algorithm where $d$ is the diameter of the graph, $m$ is the number of edges and $n$ is the number of vertices.
\end{lemma}

Ausiello, Bonifaci \& Escoffier~\cite{AusielloBonifaciEscoffier11} examine the sort of harder problems that we are interested in (rather than $\Poly{}$-time problems), however their setting requires an initially optimal solution and only covers a single perturbation, rather than the highly dynamic environment expected from the application domain we consider.

Of course there are limits to the existence of maintenance algorithms:

\begin{theorem}
For any computable function $f$ and polynomial $p$, if a problem $\Pi$ is $\W[t]$-hard for parameter $k$ and any $t \geq 1$, $\Pi$ has no $O(f(k)p(n))$-maintenance algorithm unless $\W[t] = \FPT{}$ where $n$ is the size of the input.
\end{theorem}

\begin{proof}
Assume that $\Pi$ has such an algorithm.

Then let $(G,k)$ be an instance of $\Pi$ and $G_{0}$ be the trivial instance defined by taking the completely disconnected graph on $\Card{V(G)}$ vertices with trivial optimal solution $S_{0}$. Let $G_{0}, \ldots ,G_{\Card{E(G)}}$ be a sequence of graphs generated by adding the edges of $G$ one by one in arbitrary order to $G_{0}$, resulting in $G_{\Card{E(G)}} = G$.

As $\Pi$ has a maintenance algorithm beginning with $G_{0}$ the algorithm can be applied to each pair $G_{i}$, $G_{j}$ where $j=i+1$ with solution $S_{i}$ to obtain solution $S_{j}$. Then if $S_{\Card{E(G)}}$ is a witness that $(G,k)$ is a \Yes{}-instance, the algorithm answers \Yes{} (or returns $S_{|E(G)|}$ in the case of a search problem) and \No{} otherwise.

Then we have an algorithm that performs $O(f(k)p(n))$ operations for each vertex at each iteration. As there are at most $n^{2}$ iterations, the overall algorithm has running time $O(n^{3}f(k)p(n))$, and hence $\Pi$ is fixed-parameter tractable, which contradicts the $\W[t]$-hardness of $\Pi$ with parameter $k$.
\end{proof}

By the same argument, we can also obtain a classical complexity analogue.

\begin{corollary}
For any polynomial $p$, if a problem $\Pi$ is $\NP{}$-hard, $\Pi$ has no $O(p(n))$-maintenance algorithm that optimally solves $\Pi$ unless $\Poly{} = \NP{}$ where $n$ is the size of the input.
\end{corollary}

Similarly, classical approximation results are preserved.

\begin{corollary}
For any polynomial $p$, if a problem $\Pi$ is $\NPO{}$-PB-hard, $\Pi$ has no $O(p(n))$-maintenance algorithm that solves $\Pi$ within an approximation factor of $O(n^{1-\varepsilon})$ for any $\varepsilon > 0$ unless $\Poly{} = \NP{}$ where $n$ is the size of the input.
\end{corollary}

If, given a parameterized problem with parameter $k$, the bound is $O(f(k)p(n))$ where $f$ is a computable function, $p$ is a polynomial in $n$, the size of the input, we denote the associated bounded maintenance algorithm as an ($r$-local) $\FPT{}$-maintenance algorithm. The complexity results given in Section~\ref{sec:evolving_graphs} then give the following results.

\begin{theorem}
\textsc{$g(k)$-Approx-Vertex Deletion to Regular Subgraph} has no $\FPT{}$-maintenance algorithm unless $\W[1]=\FPT{}$.
\end{theorem}

\begin{theorem}
\textsc{$g(k)$-Approx-Deletion to Regular Subgraph} has no $\FPT{}$-maintenance algorithm unless $\W[1]=\FPT{}$.
\end{theorem}

\begin{theorem}
\textsc{$g(k)$-Approx-Weighted Degree Constrained Deletion} has no $\FPT{}$-maintenance algorithm unless $\W[1]=\FPT{}$.
\end{theorem}

Similar statements can be made regarding \textsc{$c$-Add-Approx-Dominating Set} and \textsc{$g(k)$-Approx-Independent Dominating Set} using their respective approximation hardness results:

\begin{lemma}[\cite{DowneyFellowsMcCartinRosamond08}]
\sloppypar\textsc{$c$-Add-Approx-Dominating Set} has no $\FPT{}$-maintenance algorithm unless $\W[2]=\FPT{}$.
\end{lemma}

\begin{lemma}[\cite{DowneyFellowsMcCartinRosamond08}]
\textsc{$g(k)$-Approx-Independent Dominating Set} has no $\FPT{}$-maintenance algorithm unless $\W[1]=\FPT{}$.
\end{lemma}

Independently of the results above, we observe that if the bound on the maintenance algorithm is further restricted, then in certain cases no approximation ratio can be maintained. More precisely the approximation ratio will be a function of the number of iterations of the maintenance algorithm, up to trivial bounds on the quality of the solution. Given a (w.l.o.g) minimization problem $\Pi$ in $\normalfont{NPO}$ with the following properties:

\begin{enumerate}
\item The decision counterpart of $\Pi$ is $\normalfont{NP}$-hard.
\item $\Pi$ has a $O(1)$-maintenance algorithm $\mathcal{A}$.
\end{enumerate}

We make the following claims:

\begin{lemma}
There is a step in $\mathcal{A}$ that is \emph{divergent}.
\end{lemma}

Where \emph{divergent} means that if the size of the optimal solution decreases, the size of the solution given by $\mathcal{A}$ does not decrease, and similarly if the size of the optimal solution does not change, the size of the solution given by $\mathcal{A}$ increases.

\begin{proof}
Follows immediately from the $\normalfont{NP}$-hardness of the decision counterpart of $\Pi$.
\end{proof}

Let $\gamma_{i}$ be the size of the solution given as input from instance $I_{i}$ (perhaps as the result of a previous application of $\mathcal{A}$) and $\gamma^{*}_{i}$ be the size of the optimal solution for $I_{i}$. Let $I_{i+1}$ be the instance after alterations and $\gamma_{i+1}$ be the size of the solution given by $\mathcal{A}$ and $\gamma_{i+1}^{*}$ be the size of the optimal solution.

\begin{lemma}
Given an approximation ratio $A$ and an instance with solution of size $\gamma_{i} = A\cdot\gamma_{i}^{*}$, then $\mathcal{A}$ cannot guarantee that $\gamma_{i+1} \leq A\cdot\gamma_{i+1}^{*}$.
\end{lemma}

\begin{proof}
Let the change from $I_{i}$ to $I_{i+1}$ be such that it induces $\mathcal{A}$ to take a divergent step. Assume without loss of generality that $\gamma_{i+1} = \gamma_{i}+1$ and $\gamma_{i+1}^{*} = \gamma_{i}^{*}$ Then
\begin{equation*}
\frac{\gamma_{i+1}}{\gamma_{i+1}^{*}} = \frac{\gamma_{i}+1}{\gamma_{i}^{*}} = \frac{A\cdot\gamma_{i}^{*}+1}{\gamma_{i}^{*}} = A+\frac{1}{\gamma_{i}^{*}}
\end{equation*}

It is easy to see that a similar result occurs with the other possibilities for divergence.
\end{proof}

Given a problem $\Pi$ with a \emph{long divergent sequence} this problem may be amplified. A \emph{long divergent sequence} is a set of instances where the changes between each instance induce divergent steps in $\mathcal{A}$. For example given the completely disconnected graph and the problem \textsc{Dominating Set}, the initial solution is to take all vertices in the dominating set (which is optimal), then one by one we add edges to obtain a star graph. Using an adversarial approach we add the edges such that the centre of the star is not in the dominating set, but all other vertices are, taking $n-1$ vertices where only $1$ is needed.

\begin{lemma}
Given a long divergent sequence of length $d$, we have $\gamma_{d+1} / \gamma_{d+1}^{*} \geq 1+d/\gamma^{*}_{d+1}$.
\end{lemma}

\begin{proof}
At each step we are forced to increase the size of the solution relative to the optimal by at least one. Then by step $d$ (instance $I_{d+1}$) we have $\gamma_{d+1} \geq \gamma_{1} + d$ however $\gamma^{*}_{d+1} \leq \gamma^{*}_{1}$.

Therefore we at least have:
\begin{equation*}
\frac{\gamma_{d+1}}{\gamma_{d+1}^{*}} \geq \frac{\gamma_{1}+d}{\gamma_{1}^{*}} =  \frac{\gamma_{1}^{*}+d}{\gamma_{1}^{*}}
 =  1+\frac{d}{\gamma_{1}^{*}}
 \geq  1 + \frac{d}{\gamma_{d+1}^{*}}
\end{equation*}
\end{proof}

Then for the dominating set example we have the approximation ratio of $O(n)$. In general, given a problem with a long divergent sequence of length $O(n)$, we obtain a similar approximation ratio.

\section{Problems in Dynamic Graphs}
\label{sec:evolving_graphs}

Given a graph, the problem of obtaining an $r$-regular subgraph with a minimum number of excluded vertices (known variously as \textsc{$r$-Regular Subgraph}, \textsc{$k$-Almost $r$-Regular Subgraph} and \textsc{Vertex Deletion to Regular Subgraph}) is $\W[1]$-hard~\cite{MathiesonSzeider12}. The property of having a bounded degree or regular graph is particularly interesting for routing purposes. Although it is $\W[1]$-hard, for a dynamic graph, an approximation would be sufficient as we expect the graph to change rapidly, so computing a reasonable solution quickly is more effective than computing an exact solution slowly. In this section however we show that this problem has no parameterized approximation unless $\W[1]=\FPT{}$.

Formally we define the problem as:

\pcproblem{$g(k)$-Approx-Vertex Deletion to Regular Subgraph}{A graph $G=(V,E)$, integers $k$ and $r$.}{$k$.}{Is there a set $V'\subset V$ with $\Card{V'}\leq k$ such that the subgraph $G'=G[V\setminus V']$ is $r$ regular?}

This problem has many possible generalizations, we note two in particular:

\pcproblem{$g(k)$-Approx-Deletion to Regular Subgraph}{A graph $G=(V,E)$, integers $k$ and $r$.}{$k$.}{Is there a set $D\subset V\cup E$ with $\Card{D}\leq k$ such that the subgraph $G'=(V\setminus D, E\setminus D)$ is $r$ regular?}

\pcproblem{$g(k)$-Approx-Weighted Degree Constrained Deletion}{A graph $G=(V,E)$, integers $k$ and $r$, a weight function $\rho: V\cup E \rightarrow [1,k+1]$ and a degree list function $\delta: V \rightarrow 2^{\{0,\ldots,r\}}$.}{$k$.}{Is there a set $D\subset V\cup E$ with $\sum_{d \in D} \rho(d) \leq k$ such that for every vertex $v$ in the subgraph $G'=(V\setminus D, E\setminus D)$ we have $\sum_{e \in N^{G'}(v)}\rho(e) \in \delta(v)$?}

The reduction we use is from the following problem:

\pcproblem{Strongly Regular Multicolored Clique}{A graph $G=(V,E)$ with $V= \uplus_{i \in [k]} V_{i}$ such that $\Card{V_{i}} = s$ and for all $v \in V_{i}$ we have $d(v)|_{V_{j}} = d$ for all $i$, $j$, and an integer $k$.}{$k$.}{Is there a clique $V'$ of size $k$ such that $V'$ has one vertex from each $V_{i}$?}

In the remainder of this section we present the hardness reduction that demonstrates that \textsc{$g(k)$-Approx-Vertex Deletion to Regular Subgraph} and subsequently the more general problems are $\W[1]$-hard, and therefore have no $\FPT{}$-maintenance algorithms. We then prove that problems that are fixed-parameter tractable may have an approximate maintenance algorithm by proving that this is indeed the case for \textsc{Vertex Cover}.

\subsection{Hardness Results}

\begin{lemma}
\textsc{$g(k)$-Approx-Vertex Deletion to Regular Subgraph} is $\W[1]$-hard.
\end{lemma}

\begin{proof}
The underlying exact problem \textsc{Vertex Deletion to Regular Subgraph} was shown to be $\W[1]$-hard in~\cite{MathiesonSzeider12} with a reduction than can be used with little modification to show hardness for the approximation problem.

The reduction is from \textsc{Strongly Regular Multicolored Clique}. Let $(G,k)$ of \textsc{Strongly Regular Multicolored Clique} where $V(G) = \uplus_{i \in [k]} V_{i}$ is partitioned into $k$ color classes where $\Card{V_{i}} = s$ for all $i$ and each vertex has $d$ neighbours in each color class. We construct an instance $(G',k')$ of \textsc{$g(k)$-Approx-Vertex Deletion to Regular Subgraph} by setting $k' = k + \binom{k}{2}$, creating a vertex copy $V'_{i}$ for each color class $V_{i}$ and a set $P_{ij}$ of vertices for each pair $V_{i}$, $V_{j}$ of color classes, which will control edge selection. We make each $V'_{i}$ a complete graph. For every pair of vertices $u$,$v$ with $u \in V_{i}$ and $v \in V_{j}$, $i \neq j$, let $u'$ and $v'$ be the corresponding vertices in $V'_{i}$ and $V'_{j}$. If $uv \in E(G)$ we add two vertices $u'_{v'}$, $v'_{u'}$ to $P_{ij}$ with the edges $u'u'_{v'}$, $u'_{v'}v'_{u'}$ and $v'v'_{u'}$. For each pair of vertices $u_{v}$, $u'_{v'}$ in $P_{ij}$ where $u$ and $u'$ are in the same $V_{l}$ ($l =i,j$) and $u\neq u'$ we add the edge $u_{v}u'_{v'}$.
We choose $r$ to be greater than $\max\{(s-1)+d(k-1),2+(s-1)d\}$ such that $r$ is also of opposite parity to $s$, the smallest such $r$ suffices. Then for each $V'_{i}$ we add a set of vertices $V''_{i}$ where each vertex in $V''_{i}$ has an edge to each vertex in $V'_{i}$ such that the degree of each vertex in $V'_{i}$ is $r+1$. We do similarly for each $P_{ij}$ with a set $P'_{ij}$. We then increase the degree of each vertex in each $V''_{i}$ and each $P'_{ij}$ to $r+1$ by taking an set of $r+1$ clique, breaking an edge in each and attaching it at these two points to the vertex being adjusted. For each of these vertices, if the total number of vertices in the attached cliques is less than $g(k')$ we augment them by adding new $r+1$ cliques by breaking an edge in the old and new clique and reattaching them to each other.

\begin{claim}
If $G$ has a properly colored clique of size at least $k$, then we can delete at most $k'$ vertices to make $G'$ $r$-regular.
\end{claim}

This claim follows as in \cite{MathiesonSzeider12}. Note in particular that $G$ has a clique of size at most $k$ (as there are only $k$ color classes and we must use at least $k'$ vertices, one from each $V'_{i}$ and two from each $P_{ij}$ to make the graph $r$-regular. The vertices chosen for deletion correspond to the vertices and edges of the clique.

\begin{claim}
If $G'$ can be made $r$-regular by the deletion of at most $g(k')$ vertices, then $G$ has a properly colored clique of size at least $k$.
\end{claim}

Note that we must delete at least one vertex from each $V'_{i}$ and two from each $P_{ij}$. If we delete any vertices from the degree adjustment components of the graph, $V''_{i}$, $P'_{ij}$ or the adjustment cliques, then we must delete all such vertices, which would require more than $g(k)$ deletions. Therefore the $k'$ vertices must be the only vertices deleted. As with \cite{MathiesonSzeider12}, these must correspond to a properly colored $k$-clique in $G$.
\end{proof}

Note also that if we allow edge deletion as well as vertex deletion, we obtain no advantage by deleting any edges, therefore any deletions must be of vertices only. This gives the following corollary.

\begin{corollary}
\textsc{$g(k)$-Approx-Deletion to Regular Subgraph} is $\W[1]$-hard.
\end{corollary}

By subproblem containment, the more general version where we allow each vertex to have a list of possible final degrees, rather than simple degree $r$ is also $\W[1]$-hard.

\begin{corollary}
\textsc{$g(k)$-Approx-Weighted Degree Constrained Deletion} is $\W[1]$-hard when the possible editing operations are vertex deletion or vertex and edge deletion.
\end{corollary}

\subsection{A Maintenance Algorithm for Vertex Cover}

Thus we may conclude that maintenance is not feasible (in the complexity guarantee sense) even if we have polynomial time at each step (in fact we may relax this perhaps to fixed-parameter tractable time, dependent on the number of steps), if the parameterized approximation problem is $\W[t]$-hard for any $t \geq 1$.

\sloppypar However if the parameterized approximation problem is fixed-parameter tractable, then we may have some hope. In a broad sense many kernelisation algorithms may be reinterpreted as maintenance algorithms, e.g.,

\begin{lemma}[\cite{Prieto05}]
\textsc{$k$-Max Cut} has an $\FPT{}$-time $2k$-local maintenance algorithm.
\end{lemma}

\begin{lemma}[Attribute to Buss in~\cite{DowneyF99}]
\textsc{Vertex Cover} has an $\FPT{}$-time $2k^2$-local maintenance algorithm.
\end{lemma}

Note that these examples were not chosen as necessarily the most efficient algorithms for these problem, but because they employ simple, local reduction rules. However in both cases (and all such adaptations) the algorithm does not use any prior solution, so perhaps is not the most effective maintenance strategy. Nonetheless they do give goals for the performance of any maintenance algorithm.

By allowing an approximate solution, we can develop potentially much simpler, faster and more local maintenance algorithms. We demonstrate with the following case, \textsc{Vertex Cover}, which is fixed-parameter tractable, therefore fixed-parameter approximable and has a simple greedy $2$-approximation. Furthermore there is evidence that a $2$-approximation is in some sense the best possible ratio~\cite{KhotRegev08}. In this case, the following maintenance algorithm is efficient and of tolerable approximation quality.

\begin{lemma}
There is a 1-local maintenance algorithm for \textsc{Maximal Matching} and hence there is a 1-local maintenance algorithm for \textsc{Vertex Cover} with an approximation factor of $2$.
\end{lemma}

\begin{proof}
The algorithm exploits the result that any maximal matching is a $2$-approximation for the size of the minimum vertex cover.

Given a graph $G$, with vertex cover $S$, and a set of edge deletions and additions $T$, the algorithm proceeds by maintaining a matching at the site of the deletions and additions. We may assume that at each completed step in the algorithm (as this algorithm can be used to generate the initial solution) each non-isolated vertex in the vertex cover has a \emph{pair-vertex}, its partner in the matching.

Then for each element $t$ of $T$, if $t$ is a deletion that separates a vertex from its pair-vertex, and the vertex has any neighbour not in the cover, we choose arbitrarily one of these neighbours to be its new pair-vertex. If it has no such neighbours, the vertex is removed from the cover.

Otherwise if $t$ is an addition that connects two vertices not in the cover, they are both added to the cover and are noted as each others pair-vertex.

Then at any point in the algorithm each vertex in the cover has an adjacent pair-vertex (i.e., the edge between the two is a matching edge), and there is no edge that has neither endpoint in the cover. Therefore the implicit matching is maximal. Thus we immediately have our $2$-approximation.
\end{proof}

By way of contrast, \textsc{Max Cut} admits a simple $\frac{1}{2}$ approximation algorithm~\cite{SahniGonzalez76} (though this is far from the best ratio) which is easily translated into a simple maintenance algorithm however the locality seems to be the diameter of the graph in the worst case. This suggests that there is potentially some separation possible between ``local'' and ``non-local'' problems that is not captured by regular complexity classes. This idea has been explored from a different perspective by Milteren \emph{et al.}~\cite{Miltersen1994} who show a split of $\Poly{}$-complete problems by the class $\normalfont{incr\text{-}POLYLOGTIME}$.

\section{Conclusion}
\label{sect:conclusion}


We show that $\W[1]$-hardness precludes maintenance, and indeed given sufficient restriction on the stepwise computation no maintenance algorithm can guarantee any but the most trivial approximation ratios. However, as demonstrated with \textsc{Vertex Cover}, some problems do admit fast, effective approximate maintenance algorithms. It is not clear whether there is a neat classification of these problems, however fixed-parameter tractability for the related parameterized (approximation) problem(s) is at least a necessary condition. The most likely situation is that there is some proper subset of the problems in $\FPT{}$ that admit a maintenance algorithm, as the structure of a maintenance algorithm seems to rely on some as yet undefined localizing property which is unlikely to be identical to the property of being fixed-parameter tractable. If indeed there is a coherent underlying property that allows maintenance, then perhaps it would also be possible to construct \emph{maintenance-preserving} reductions (similar to the \emph{incremental reductions} of \cite{Miltersen1994}), which would allow a simpler and more robust demonstration of the possibility or impossibility of maintenance for a given problem.

\bibliographystyle{abbrv}

\appendix

\section{Parameterized Approximation Complexity}
\label{sect:para}

In this section, we only provide the necessary definitions and properties of Parameterized Complexity and Parameterized Approximation Complexity. We refer the reader to~\cite{DowneyF99,FlumG06} and~\cite{Marx08} respectively for more details.

\subsection{Basic Parameterized Complexity}

Parameterized Complexity explores the complexity of combinatorial problems using parameters as independent measures of structure in addition to the overall size of the input. An instance $(I,k)$ of a parameterized problem consists of the input $I$, corresponding to the input of a classical problem and an integer parameter $k$, a special part of the input independent from $\Card{I}$.

A problem is \emph{fixed-parameter tractable} (or in $\FPT{}$) if there is an algorithm that solves each instance $(I,k)$ of the problem in time bounded by $f(k)\cdot\Card{I}^{O(1)}$ where $f$ is a computable function dependent only on $k$.

Conversely hardness for any class in the ${\normalfont W}$-hierarchy provides evidence that a problem is not fixed-parameter tractable. Hardness for such classes is typically established by \emph{parameterized reduction}, the Parameterized Complexity reduction scheme where given an instance $(I,k)$ of problem $\Pi_{1}$ an instance $(I',k')$ of problem $\Pi_{2}$ is computed in time bounded by $f(k)\cdot\Card{I}^{O(1)}$, with $k' \leq g(k')$ for some computable function $g$ and $(I,k)$ is a \Yes{}-instance if and only if $(I',k')$ is a \Yes{}-instance.

\subsection{Parameterized Approximation Complexity}

The additional measure embodied in the parameter also provides an alternative possibility for approximation problems. Of particular interest is approximating the cost of the solution, where the parameter $k$ is the desired cost of the solution. Given parameterized problem $\Pi$ with an additional optimization objective (either $\min$ or $\max$), the general cost minimization (maximization) parameterized problem for $\Pi$ is:

\pcaproblem{$g(k)$-Approx-$\Pi$}{A instance $I$ of $\Pi$, an integer $k$.}{$k$.}{Either \No{}, asserting that there is no solution of size at most (at least) $k$ for $I$, or a solution of size at most (at least) $g(k)$.}

Of course for minimization problems the approximation is only interesting if $g(k) \geq k$, and vice versa for maximization. For specific functions $g(k)$ we obtain the following two interesting subcases:

\pcaproblem{$c$-Add-Approx-$\Pi$}{A instance $I$ of $\Pi$, an integer $k$.}{$k$.}{Either \No{}, asserting that there is no solution of size at most (at least) $k$ for $I$, or a solution of size at most $k+c$ (at least $k-c$).}

\clearpage

\pcaproblem{$c$-Mult-Approx-$\Pi$}{A instance $I$ of $\Pi$, an integer $k$.}{$k$.}{Either \No{}, asserting that there is no solution of size at most (at least) $k$ for $I$, or a solution of size at most $ck$ (at least $k/c$).}

Then $\W[t]$-hardness for any $t \geq 1$ for such an approximation problem gives evidence that there is no fixed-parameter tractable algorithm that can produce a solution of size bounded by $g(k)$ with $k$ as the parameter.

\end{document}